\documentclass[12pt]{article}
\usepackage[english]{babel}
\usepackage{amsfonts,amsmath,amsthm,graphicx,a4wide}
\usepackage[colorlinks=true]{hyperref}

\newtheorem{thm}{Theorem}

\newcommand{\tr}[1]{\text{Tr}\left(#1\right)}
\newcommand{\trr}[1]{\text{Tr}^2\left(#1\right)}
\newcommand{\trp}[2]{\text{Tr}_{#1}\left(#2\right)}

\newcommand{\CC}{{\mathbb C}}
\newcommand{\MM}{{\mathcal M}}
\newcommand{\KK}{{\mathcal K}}

\newcommand{\EE}[1]{{\text{\large$\mathbb E$}}\left(#1\right)}
\newcommand{\bra}[1]{\langle #1 |}
\newcommand{\ket}[1]{| #1 \rangle}
\newcommand{\bket}[1]{\langle #1 \rangle}

\begin{document}

\title{Fidelity is a sub-martingale for discrete-time quantum filters}

\author{
  Pierre Rouchon\thanks{Mines ParisTech: {\tt pierre.rouchon@mines-paristech.fr}}
       }
\date{March 2011}
\maketitle
\begin{abstract}
Fidelity is known to  increase through  any Kraus map: the fidelity between two density matrices  is less than the fidelity between their images via a Kraus map.   We prove here that, in average, fidelity is also increasing for  discrete-time quantum filters attached to an arbitrary  Kraus map:  fidelity between the density matrix of the underlying Markov chain and the density matrix of the  associated  quantum filter is a sub-martingale. This result is not restricted to pure states. It also holds true  for mixed states.
\end{abstract}

\section{Introduction}

Quantum filtering  was developed, in its modern theoretical form, by Belavkin during the 1980's~\cite{belavkin-92} (see also~\cite{bouten-et-al:siam2007} for a recent tutorial introduction).  Quantum filtering  proposes a theory  for statistical inference  strongly inspired from  quantum optical systems. These systems are  described by continuous-time quantum stochastic differential equations (stochastic master equations). The stability of the obtained quantum filters, i.e, the fact that quantum filters are asymptotically independent of their initial states, has been investigated in several papers. In~\cite{van-handel:proba2009}   sufficient convergence conditions are established: they are  related to  observability issues.\footnote{See~\cite{van-handel-mtns2010}  for  general connections between the stability of nonlinear filters and mathematical systems theory concepts such as stability, observability and detectability.} In~\cite{diosi-et-al:JPhA:2006} convergence of the  estimate  to the physical state is addressed for a generic Hamiltonian and measurement operator: the analysis   relies on the fact that,  for pure states, the fidelity between the real state and its estimate is proved  to be a sub-martingale.

Quantum filtering is not restricted to  continuous time. It can also be used for  discrete-time systems as in~\cite{dotsenko-et-al:PRA09} where a quantum filter is used inside a state-feedback scheme.  This note proves  that  fidelity between the estimate and real state is always a sub-martingale for such discrete-time quantum filters. This does not mean that quantum filters are always convergent. Nevertheless,   this sub-martingale combined with stochastic invariance  arguments could be useful  to prove convergence in specific situations.

Fidelity $F(\sigma,\rho)$ between two  quantum states described by density matrices $\rho$ and $\sigma$ coincides with  their Frobenius inner product $\tr{\rho\sigma}$ when at least one of two density matrices $\rho$ or $\sigma$  is a projector of rank one, i.e. a pure state.\footnote{There are two closely related definitions of the fidelity.  The "mathematical" definition used, e.g., in~\cite{nielsen-chang-book} corresponds to the square root of the "physical" definition we have adopted in this paper for simplicity sakes.} The fact that fidelity is a sub martingale for pure states is  a known fact.  For continuous-time systems, the convergence investigations  proposed in~\cite{diosi-et-al:JPhA:2006} relies on this fact. For discrete-time systems,
theorem~3 in~\cite{mirrahimi-et-al:cdc09}, based only on Cauchy-Schwartz inequalities, proves   that the Frobenius inner product between the estimate and real state is always a sub-martingale  whatever the purity of the estimate and real state are.  Consequently, if  we assume that the  real state remains pure, then, fidelity coincides with Frobenius inner product and thus is  known to be a sub-martingale. When  the real state and its estimate  are  mixed,  fidelity does not coincide  with Frobenius inner product and  the proof that it is a sub-martingale does not result directly from published papers, as far as we know. Fidelity between  the  arbitrary  states $\rho$ and $\sigma$ is  given  by $\trr{\sqrt{\sqrt{\rho}\sigma\sqrt{\rho}}}$, an expression  much more difficult to manipulate than $\tr{\rho\sigma}$. The  contribution of this note  (theorems~\ref{thm:Qfilter} and~\ref{thm:Qfilter:ext1} relying on   inequalities~\eqref{eq:ineq} and~\eqref{eq:ineq:ext1}, respectively)  consists in proving that, in the discrete-time case and for arbitrary purity,  the fidelity  between the estimate and real state is  sub-martingale.

In section~\ref{sec:kraus}, we recall the Kraus map associated to any quantum channel,  the associated Markov chain (quantum Monte-Carlo trajectories)  and highlight the  basic inequalities~\eqref{eq:ineq} and~\eqref{eq:ineq:ext1} underlying theorems~\ref{thm:Qfilter} and~\ref{thm:Qfilter:ext1} proved in the remaining sections.
In section~\ref{sec:standard},  we consider the quantum filter attached to a quantum trajectory and prove  theorem~\ref{thm:Qfilter}: fidelity is a sub-martingale. In section~\ref{sec:aggreg} we extend this  result to a family of Markov chains attached to the same  Kraus map. In section~\ref{sec:conclusion}, we conclude by noting that, contrarily to fidelity,
trace distance and  relative entropy are not always super-martingales.

The author  thanks  Mazyar Mirrahimi and Ramon Van Handel for fruitful comments and interesting references.

\section{Kraus maps and quantum Markov chains}\label{sec:kraus}
Take the Hilbert space $S=\CC^n$ of dimension $n>0$ and consider  a quantum channel described by the  Kraus map (see~\cite{nielsen-chang-book}, chapter~8  or~\cite{haroche-raimond:book06}, chapter 4)
\begin{equation}\label{eq:kraus}
\KK(\rho) = \sum_{\mu=1}^{m} M_\mu \rho M_\mu^\dag
\end{equation}
where
\begin{itemize}
\item $\rho$ is the density matrix describing the input  quantum state, $\KK(\rho)$ being then the output quantum state;  $\rho \in \CC^{n\times n}$ is  a density matrix, i.e., an Hermitian matrix semi-positive definite and of trace one;

\item for each $\mu\in\{1,\ldots,m\}$, $M_\mu \in\CC^{n\times n}/\{0\}$, and $\sum_\mu M_\mu^\dag M_\mu= I$.
\end{itemize}
To this quantum channel is associated the following  discrete-time  Markov chain (quantum Monte-Carlo trajectories, see, e.g., \cite{haroche-raimond:book06}):
\begin{equation}\label{eq:markov}
\rho_{k+1}= \MM_{\mu_k}(\rho_k)
\end{equation}
where
\begin{itemize}
  \item $\rho_k$ is the quantum state at sampling time $t_k$ and $k$ the sampling index ($t_k < t_{k+1}$).
  \item $\mu_k\in\{1,\ldots, m\}$ is a random variable; $\mu_k=\mu$
  with  probability $p_\mu(\rho_k)=\tr{M_\mu\rho_k M_\mu^\dag}$.

\item  $\MM_\mu(\rho) = \tfrac{1}{\tr{M_\mu\rho M_\mu^\dag}} M_\mu \rho M_\mu^\dag =
     \tfrac{1}{p_\mu(\rho)} M_\mu \rho M_\mu^\dag $
     .
\end{itemize}

Quantum probabilities derive from density matrices which are
dual (via the Frobenius Hermitian product $\tr{AB^\dag}$,  $A $, $B$,  $n\times n$ matrices with complex entries)  to the set of Hermitian $n\times n$ matrices (space of observables). The
Kraus map~\eqref{eq:kraus} leads to a generalized observable (or effect
valued measure) ${\mathcal P} \mapsto \sum_{\mu \in{\mathcal P}} M_\mu^\dag M_\mu$ where $\mathcal P$ is any subset of $\{1, \ldots,m\}$. The map $({\mathcal P},\rho) \mapsto \sum_{\mu \in{\mathcal P}} M_\mu \rho M_\mu^\dag$ is known as Davies instrument (see~chapter 4 of \cite{davies:book1976}). It formalizes  a measurement notion  relating the state (the density matrix  $\rho$),  measure properties  (here $\mathcal P$ living in the $\sigma$-algebra  generated by the singletons of the discrete set $\{1,\ldots,m\}$) and  an output state conditional on the observed values (here $\sum_{\mu \in{\mathcal P}} M_\mu\rho M_\mu^\dag$ up to a normalization to ensure a unit trace). The non-linear operators $\MM_\mu$ just define  the one-step transition mechanism for the above quantum Markov chain. Such setting can be seen as  an appropriate analogue of what one encounters in classical discrete-time filtering problems.

Fidelity is one of several ways to measure difference
between density matrices and is specific to quantum theory.
The Kraus map tends  to increase fidelity $F$ (see~\cite{nielsen-chang-book}, theorem 9.6, page 414): for all density matrices $\rho$ and $\sigma$, one  has
\begin{equation}\label{eq:fidel}
 \trr{\sqrt{\sqrt{\KK(\sigma)} \KK(\rho) \sqrt{\KK(\sigma)}}}=F(\KK(\sigma),\KK(\rho)) \geq F(\sigma,\rho)=\trr{\sqrt{\sqrt{\sigma} \rho\sqrt{\sigma}}}
\end{equation}
where, for any Hermitian  semi-positive matrix $A=U\Lambda U^\dag $, $U$ unitary matrix and $\Lambda=\text{diag}\{\lambda_l\}_{l\in\{1,\ldots,n\}}$, $\sqrt{A} = U \sqrt{\Lambda} U^\dag$ with
$\sqrt{\Lambda}=\text{diag}\{\sqrt{\lambda_l}\}_{l\in\{1,\ldots,n\}}$.

The conditional expectation of $\rho_{k+1}$ knowing $\rho_k$ is given by the Kraus map:
$$
\EE{\rho_{k+1} / \rho_k} = \KK(\rho_k).
$$
This result from the trivial identity
$
\sum_{\mu=1}^{m}
  \tr{M_\mu \rho M_\mu^\dag}
     \tfrac{M_\mu \rho M_\mu^\dag }{\tr{M_\mu \rho M_\mu^\dag}},
   = \KK(\rho)
   .
$
In section~\ref{sec:standard}, we  show during the proof of theorem~\eqref{thm:Qfilter}
the following inequality
\begin{align}
\sum_{\mu=1}^{m}
  \tr{M_\mu \rho M_\mu^\dag}
  F \left(
     \tfrac{M_\mu \sigma M_\mu^\dag }{\tr{M_\mu \sigma M_\mu^\dag}},
      \tfrac{M_\mu\rho M_\mu^\dag }{\tr{M_\mu\rho M_\mu^\dag}}
     \right)
     &\geq
      F( \sigma,\rho)
      \label{eq:ineq}
\end{align}
for any density matrices $\rho$ and $\sigma$.  The left-hand side is related to a conditional expectation. Inequality~\eqref{eq:ineq}, attached to the probabilistic mapping~\eqref{eq:markov}, can be seen as   the  stochastic counter-part of inequality~\eqref{eq:fidel} attached to  the deterministic mapping~\eqref{eq:kraus}. When for some $\mu$,  $\tr{M_\mu \sigma M_\mu^\dag}=0$ with $\tr{M_\mu \rho M_\mu^\dag}>0$, the $\mu$-term in the sum~\eqref{eq:ineq} is  not defined. This is not problematic, since in this case,  if we replace
$\tfrac{M_\mu \sigma M_\mu^\dag }{\tr{M_\mu \sigma M_\mu^\dag}}$ by $\tfrac{M_\mu \xi M_\mu^\dag }{\tr{M_\mu \xi M_\mu^\dag}}$ where $\xi$ is any density matrix such that ${\tr{M_\mu \xi M_\mu^\dag}} >0$,  this term  is then well defined (in a multi-valued way)  and inequality~\eqref{eq:ineq}  remains  satisfied for any such $\xi$.

During the proof of theorem~\eqref{eq:Qfilter:ext1}, we extend this inequality to any  partition of $\{1,\ldots,m\}$ into $p\geq 1$ sub-sets ${\mathcal P}_\nu$:
\begin{align}
\sum_{\nu=1}^{\nu=p}
  \tr{\sum_{\mu\in{\mathcal P}_{\nu}}M_\mu \rho M_\mu^\dag}
  F \left(
     \tfrac{\sum_{\mu\in{\mathcal P}_{\nu}}M_\mu \sigma M_\mu^\dag }{\tr{\sum_{\mu\in{\mathcal P}_{\nu}}M_\mu \sigma M_\mu^\dag}},
      \tfrac{\sum_{\mu\in{\mathcal P}_{\nu}}M_\mu\rho M_\mu^\dag }{\tr{\sum_{\mu\in{\mathcal P}_{\nu}}M_\mu\rho M_\mu^\dag}}
     \right)
     &\geq
      F( \sigma,\rho)
      \label{eq:ineq:ext1}
\end{align}

\section{The standard case.}\label{sec:standard}
 Take  a realization of the Markov chain associated to the Kraus map $\KK$.  Assume that we detect, for each $k$, the jump $\mu_k$  but   that we do not know the   initial state $\rho_0$. The objective is to propose at sampling $k$,  an estimation $\hat\rho_k$ of $\rho_k$ based on the past detections $\mu_0,\ldots,\mu_{k-1}$. The simplest method consists in starting from   an initial estimation $\hat \rho_0$  and at each sampling step to jump according to the detection. This leads to  the following estimation scheme  known as {\em a quantum filter} (see, e.g.,~\cite{wiseman-milburn:book}):
\begin{equation}\label{eq:Qfilter}
\hat \rho_{k+1}= \MM_{\mu_k}(\hat \rho_k)
\end{equation}
with   $p_\mu(\rho_k)=\tr{M_\mu\rho_k M_\mu}$  as probability of  $\mu_k=\mu$. Notice that when $\tr{M_{\mu_k}\hat \rho_k M_{\mu_k}}=0$, $\MM_{\mu_k}(\hat \rho_k)$ is not defined and should be replaced by $\MM_{\mu_k}(\xi)$ where $\xi$ is any density matrix such that $\tr{M_{\mu_k}\hat \xi M_{\mu_k}} >0$ (take, e.g.,  $\xi=\tfrac{1}{n} I_d$). The theorem here below could be useful to investigate  the  convergence of $\hat\rho_k$ towards $\rho_k$ as $k$ increases.
\begin{thm}\label{thm:Qfilter}
Consider the Markov chain of state $(\rho_k,\hat\rho_k)$ satisfying~\eqref{eq:markov} and~\eqref{eq:Qfilter}. Then $F(\hat\rho_k,\rho_k)$ is a sub-martingale:
$
\EE{F(\hat\rho_{k+1},\rho_{k+1}) / (\hat \rho_k,\rho_k) } \geq F(\hat\rho_{k},\rho_{k})
.
$
\end{thm}

The proof proposed here below deals with the general case when both $\rho_k$ and $\hat\rho_k$ can be mixed states.  It  relies on arguments similar to those used for the proof of  theorem 9.6 in~\cite{nielsen-chang-book}. They are  based on a kind of lifting process, usual in differential geometry. They consist  in introducing additional variables and embedding the Hilbert space $S$ into the tensor product with a copy $Q$ of $S$ and with the environment-model space $E$. In the "augmented  space" $S\otimes Q\otimes E$, the quantum dynamics~\eqref{eq:markov} read as an unitary transformation followed by an orthogonal projection. At sampling step $k$, the estimated and real states are  lifted to the "augmented space" via a standard procedure known as purification. These purifications are then subject to an  unitary transformation followed by  the same  projection.   Partial traces over $Q\otimes E$ provide then the  estimated and real states  at sampling time~$k+1$. At each sampling step $k$, the purifications of the estimated and real states are chosen according to Ulhmann's theorem: the square of modulus of  their Hermitian  product is maximum and coincides thus with the fidelity. The remaining part of the proof is based on simple manipulations of Hermitian products  and  Cauchy-Schwartz inequalities.

\begin{proof} The density matrices $\rho$ and $\hat\rho$ are operators from $S=\CC^{n}$ to $S$. Take a copy $Q=\CC^n$ of $S$ and consider  the composite system living on $S\otimes Q \equiv \CC^{n^2}$. Then $\hat\rho$ and $\rho$  correspond to partial traces versus $Q$ of projectors $\ket{\hat\psi}\bra{\hat\psi}$ and  $\ket{\psi}\bra{\psi}$ associated to pure states $\ket{\hat\psi}$ and $\ket{\psi} \in S\otimes Q$:
 $$
 \hat\rho =\trp{Q}{\ket{\hat\psi}\bra{\hat\psi}}, \quad  \rho =\trp{Q}{\ket{\psi}\bra{\psi}}
 $$
$\ket{\hat\psi}$ and $\ket{\psi}$ are called  purifications of $\hat\rho$ and $\rho$. They are not unique but one can always choose them such that $F(\hat\rho,\rho)= |\bket{\hat\psi| \psi}|^2$ (Uhlmann's theorem, see~\cite{nielsen-chang-book}, theorem 9.4 page 410).
Denote by  $\ket{\hat \psi_k}$ and $\ket{\psi_k}$   such purifications of $\hat\rho_k$ and $\rho_k$  satisfying
$$
F(\hat\rho_k,\rho_k)= |\bket{\hat\psi_k| \psi_k}|^2
.
$$

We have
$$
  \EE{F(\hat\rho_{k+1},\rho_{k+1}) / (\hat \rho_k,\rho_k) }
   = \sum_{\mu=1}^{m}
   p_\mu(\rho_k) F(\MM_\mu(\hat \rho_k),\MM_\mu(\rho_k))
   .
$$
The matrices $\MM_\mu(\hat \rho_k)$ and $\MM_\mu(\rho_k)$ are also density matrices.

 We will consider now the space $S\otimes Q \otimes E$ where $E=\CC^{m}$ is the Hilbert space of the environment appearing in the system-environment model of the Kraus map~\eqref{eq:kraus} (see~\cite{haroche-raimond:book06}, chapter 4 entitled "The environment is watching"). Thus  exist $\ket{e_0}$ a unitary  element of $E$,  an  ortho-normal basis  $(\ket{\mu})_{\mu\in\{1,\ldots,m\}}$ of $E$ and  a unitary transformation $U$ (not unique) on
$S\otimes E$ such that, for all $\ket{\phi}\in S$,
$$
U\left(\ket{\phi} \otimes \ket{e_0}\right) = \sum_{\mu=1}^{m}  \left(M_\mu \ket{\phi}\right) \otimes \ket{\mu}
.
$$
For each $\mu$, denote by $P_\mu$ the orthogonal projector onto  the subspace $S\otimes (\CC\ket{\mu})$. Then $P_\mu U \left(\ket{\phi} \otimes \ket{e_0}\right) = \left(M_\mu \ket{\phi}\right) \otimes \ket{\mu}$ and $\sum_\mu P_\mu = I$. For any density matrix $\rho$ on $S$,
$$
P_\mu U \left(\rho \otimes \ket{e_0}\bra{e_0} \right)U^\dag P_\mu
= M_\mu \rho M_\mu^\dag \otimes \ket{\mu}\bra{\mu}
$$
and thus
$$
\trp{E}{P_\mu U \left(\rho \otimes \ket{e_0}\bra{e_0} \right)U^\dag P_\mu}
= M_\mu \rho M_\mu^\dag
.
$$

 The   unitary transformation $U$ of $S\otimes E$ can be  extended  to $ S\otimes Q \otimes E \equiv S\otimes E \otimes Q$ by setting  $V=U\otimes I$ ($I$ is identity on $Q$). Since $\ket{\psi_k}\in S\otimes Q$ is a purification of $\rho_k$, we have
$$
M_\mu \rho_k M_\mu^\dag = \trp{E\otimes Q}{ P_\mu V \left(\ket{\psi_k}\bra{\psi_k}\otimes \ket{e_0}\bra{e_0}\right) V^\dag P_\mu}
.
$$
Set $\ket{\phi_k}=\ket{\psi_k}\otimes\ket{e_0} \in S\otimes Q\otimes E$ and $\ket{\chi_k}= V \ket{\phi_k}$. Using $P_\mu^2=P_\mu$,  we have
$$
p_\mu(\rho_k)
=\tr{M_\mu \rho_k M_\mu^\dag }= \bket{\phi_k|V^\dag P_\mu V|\phi_k} = \|P_\mu \ket{\chi_k}\|^2.
$$
Thus,  for each $\mu$, the  state $
\ket{\chi_{k\mu}}=\tfrac{1}{\sqrt{p_\mu(\rho_k)}}P_\mu  \ket{\chi_k}$
is a purification of $\MM_\mu(\rho_k)$:
$$
\MM_\mu(\rho_k)= \trp{Q\otimes E}{\ket{\chi_{k\mu}}\bra{\chi_{k\mu}}}.
$$
Similarly set  $\ket{\hat \phi_k}=\ket{\hat\psi_k}\otimes\ket{e_0}$ and $\ket{\hat \chi_k}= V \ket{\hat \phi_k}$.  For each $\mu$,
$
\ket{\hat \chi_{k\mu}}=\tfrac{1}{\sqrt{p_\mu(\hat\rho_k)}} P_\mu \ket{\hat\chi_{k}}
$
is  also a purification of $\MM_\mu(\hat\rho_k)$. By Uhlmann's theorem,
$$
F(\MM_\mu(\hat \rho_k),\MM_\mu(\rho_k)) \geq |\bket{\hat\chi_{k\mu}| \chi_{k\mu}} |^2
.
$$
Thus we have
$$
\EE{F(\hat\rho_{k+1},\rho_{k+1}) / (\hat \rho_k,\rho_k) }
   \geq  \sum_{\mu=1}^{m}  p_\mu(\rho_k)~ |\bket{\hat\chi_{k\mu}| \chi_{k\mu}} |^2
   .
$$
Since $V$ is unitary,
$$
|\bket{\hat \chi_k|\chi_k}|^2 = |\bket{\hat \phi_k|\phi_k}|^2= |\bket{\hat \psi_k|\psi_k}|^2= F(\hat\rho_k,\rho_k)
.
$$
Let us show  that
$
\sum_{\mu=1}^{m}  p_\mu(\rho_k)~ |\bket{\hat\chi_{k\mu}| \chi_{k\mu}} |^2 \geq  |\bket{\hat \chi_k|\chi_k}|^2
.
$
We have $$ p_\mu(\rho_k)~ |\bket{\hat\chi_{k\mu}| \chi_{k\mu}}|^2= |\bket{\hat\chi_{k\mu}| P_\mu \chi_k}|^2=
|\bket{\hat\chi_{k\mu}|\chi_k}|^2,$$ thus it is enough  to prove that
$
\sum_{\mu=1}^{m} |\bket{\hat\chi_{k\mu}| \chi_k}|^2  \geq  |\bket{\hat \chi_k|\chi_k}|^2
$.
Denote by  $\hat R\subset S\otimes Q\otimes E$ the vector space spanned by the ortho-normal basis $\left(\ket{\hat\chi_{k\mu}}\right)_{\mu\in\{1,\ldots,m\}}$ and  by $\hat P$ the projector on $\hat R$.  Since
$$\ket{\hat\chi_k} = \sum_{\mu=1}^{m} P_\mu \ket{\hat\chi_k} = \sum_{\mu=1}^{m}  \sqrt{p_\mu(\hat \rho_k)} \ket{\hat\chi_{k\mu}}$$
$\ket{\hat\chi_k}$ belongs to $\hat R$ and thus $|\bket{\hat\chi_k| \chi_k}|^2=|\bket{\hat \chi_k | \hat P\ket{\chi_k}}|^2$.
We conclude by  Cauchy-Schwartz inequality
$$
|\bket{\hat\chi_k| \chi_k}|^2=|\bket{\hat \chi_k | \hat P\ket{\chi_k}}|^2 \leq \|\hat \chi_k \|^2 \|\hat P\ket{\chi_k}\|^2= \|\hat P\ket{\chi_k}\|^2=\sum_{\mu=1}^{m} |\bket{\hat\chi_{k\mu}|\chi_k}|^2.
$$

\end{proof}

\section{The coarse-grained  case.}\label{sec:aggreg}
Let us consider another    Markov chain  associated to a partition of $\{1,\ldots,m\}$ into $p\geq 1$ sub-sets ${\mathcal P}_\nu$ (aggregation  of several quantum jumps via "partial Kraus maps"):
\begin{equation}\label{eq:markov:ext1}
\rho_{k+1} = \tfrac{1}{\tr{\sum_{\mu\in{\mathcal P}_{\nu_k}} M_\mu \rho_k M_\mu^\dag }}\left( \sum_{\mu\in{\mathcal P}_{\nu_k}} M_\mu \rho_k M_\mu^\dag \right)
\end{equation}
where $\nu_k=\nu$ with probability $\tr{\sum_{\mu\in{\mathcal P}_\nu} M_\mu \rho_k M_\mu^\dag }$. It is clear that
$\EE{\rho_{k+1}/\rho_k} = \KK(\rho_k)$: the Markov chains~\eqref{eq:markov} and~\eqref{eq:markov:ext1} are  attached to the same Kraus map~\eqref{eq:kraus}.
Consider the associated quantum filter
\begin{equation}\label{eq:Qfilter:ext1}
\hat \rho_{k+1} = \tfrac{1}{\tr{\sum_{\mu\in{\mathcal P}_{\nu_k}} M_\mu \hat \rho_k M_\mu^\dag }}\left( \sum_{\mu\in{\mathcal P}_{\nu_k}} M_\mu \hat \rho_k M_\mu^\dag \right)
\end{equation}
where the jump index $\nu_k$ coincides with the jump index $\nu_k$ in~\eqref{eq:markov:ext1}.
Then we have the following theorem.
\begin{thm}\label{thm:Qfilter:ext1}
Consider the Markov chain of state $(\rho_k,\hat\rho_k)$ satisfying~\eqref{eq:markov:ext1} and~\eqref{eq:Qfilter:ext1}. Then $F(\hat\rho_k,\rho_k)$ is a sub-martingale:
$
\EE{F(\hat\rho_{k+1},\rho_{k+1}) / (\hat \rho_k,\rho_k) } \geq F(\hat\rho_{k},\rho_{k})
.
$
\end{thm}
\begin{proof} It is similar to the proof of theorem~\ref{thm:Qfilter}. We will just point out here the main changes using the same notations.
We start from
$$
  \EE{F(\hat\rho_{k+1},\rho_{k+1}) / (\hat \rho_k,\rho_k) }
   = \sum_{\nu=1}^{p}
   \tilde p_\nu(\rho_k) F(\tilde{\MM}_\nu(\hat \rho_k),\tilde{\MM}_\nu(\rho_k))
   .
$$
where  we have set
$$
\tilde p_\nu (\rho)= \tr{\sum_{\mu\in{\mathcal P}_{\nu}} M_\mu  \rho M_\mu^\dag },
\quad
\tilde{\MM}_\nu(\rho)= \tfrac{1}{\tilde p_\nu(\rho)} \left( \sum_{\mu\in{\mathcal P}_{\nu}} M_\mu  \rho M_\mu^\dag \right).
$$
With $\tilde P_\nu$ the orthogonal projector on  $S\otimes Q \otimes \text{span}\{\ket{\mu}, \mu\in {\mathcal P}_{\nu} \}$ and
$\tilde{M}_\nu(\rho)=  \sum_{\mu\in{\mathcal P}_{\nu}} M_\mu  \rho M_\mu^\dag$,
we have
$$
\sum_{\mu\in{\mathcal P}_{\nu}} M_\mu \hat \rho_k M_\mu^\dag
= \trp{Q\otimes E}{ \tilde{P}_\nu V \left(\ket{\psi_k}\bra{\psi_k}\otimes \ket{e_0}\bra{e_0}\right) V^\dag \tilde P_\nu}
$$
and
$$
\tilde p_\nu(\rho_k)
=\tr{\sum_{\mu\in{\mathcal P}_{\nu}} M_\mu \hat \rho_k M_\mu^\dag }= \bket{\phi_k|V^\dag \tilde P_\nu V|\phi_k} = \|\tilde P_\nu \ket{\chi_k}\|^2
$$ For each $\nu$, the  state $
\ket{\tilde\chi_{k\nu}}=\tfrac{1}{\sqrt{\tilde p_\nu(\rho_k)}}\tilde P_\nu  \ket{\chi_k}$
is a purification of $\tilde{\MM}_\nu(\rho_k)$:
$$
\tilde{\MM}_\nu(\rho_k)= \trp{Q\otimes E}{\ket{\tilde\chi_{k\nu}}\bra{\tilde \chi_{k\nu}}}.
$$
Similarly
$
\ket{\hat{\tilde{\chi}}_{k\nu}}=\tfrac{1}{\sqrt{\tilde p_\nu(\hat\rho_k)}} \tilde P_\nu \ket{\hat\chi_{k}}
$
is  also a purification of $\tilde{\MM}_\nu(\hat\rho_k)$. By Uhlmann's theorem,
$$
F(\tilde{\MM}_\nu(\hat \rho_k),\tilde{\MM}_\nu(\rho_k)) \geq |\bket{\hat{\tilde{\chi}}_{k\nu}| \tilde\chi_{k\nu}} |^2
.
$$
Thus we have
$$
\EE{F(\hat\rho_{k+1},\rho_{k+1}) / (\hat \rho_k,\rho_k) }
   \geq  \sum_{\nu=1}^{p}  \tilde p_\nu(\rho_k)~ |\bket{\hat{\tilde{\chi}}_{k\mu}| \tilde\chi_{k\mu}} |^2
   .
$$
Let us show  that
$
\sum_{\nu=1}^{p}  \tilde p_\nu(\rho_k)~ |\bket{\hat{\tilde{\chi}}_{k\nu}| \tilde\chi_{k\nu}} |^2 \geq  |\bket{\hat \chi_k|\chi_k}|^2=F(\hat\rho_k,\rho_k)
.
$
We have $$ \tilde p_\nu(\rho_k)~ |\bket{\hat{\tilde{\chi}}_{k\nu}| \tilde\chi_{k\nu}}|^2= |\bket{\hat{\tilde{\chi}}_{k\nu}| \tilde P_\nu \chi_k}|^2=
|\bket{\hat{\tilde{\chi}}_{k\nu}|\chi_k}|^2,$$ thus it is enough  to prove that
$
\sum_{\nu=1}^{p} |\bket{\hat{\tilde{\chi}}_{k\nu}| \chi_k}|^2  \geq  |\bket{\hat \chi_k|\chi_k}|^2
$.
Denote by  $\hat{\tilde{R}}\subset S\otimes Q\otimes E$ the vector space spanned by the ortho-normal basis $\left(\ket{\hat{\tilde{\chi}}_{k\nu}}\right)_{\nu\in\{1,\ldots,p\}}$ and  by $\hat{\tilde{P}}$ the projector on $\hat{\tilde{R}}$.  Since
$$\ket{\hat\chi_k} = \sum_{\nu=1}^{p} \tilde P_\nu \ket{\hat\chi_k} = \sum_{\nu=1}^{p}  \sqrt{\tilde p_\nu(\hat \rho_k)} \ket{\hat{\tilde{\chi}}_{k\nu}}$$
$\ket{\hat\chi_k}$ belongs to $\hat{\tilde{R}}$ and thus $|\bket{\hat\chi_k| \chi_k}|^2=|\bket{\hat \chi_k | \hat{\tilde{P}}\ket{\chi_k}}|^2$.
We conclude by  Cauchy-Schwartz inequality
$$
|\bket{\hat\chi_k| \chi_k}|^2=|\bket{\hat \chi_k | \hat{\tilde{P}}\ket{\chi_k}}|^2 \leq \|\hat \chi_k \|^2 \|\hat{\tilde{P}}\ket{\chi_k}\|^2= \|\hat{\tilde{P}}\ket{\chi_k}\|^2=\sum_{\nu=1}^{p} |\bket{\hat{\tilde{\chi}}_{k\nu}|\chi_k}|^2.
$$

\end{proof}

\section{Concluding remarks}\label{sec:conclusion}

Theorems~\ref{thm:Qfilter} and~\ref{thm:Qfilter:ext1} are   still valid if the Kraus operators $M_\mu$ depend on $k$. In particular, $F(\hat \rho_k,\rho_k)$ remains a sub-martingale even if the Kraus operators depend on $\hat\rho_k$, i.e.,  in case  of feedback.
 A natural extension will be to prove that  fidelity yields  also to a sub-martingale in the continuous-time case.  The Markov chain~\eqref{eq:markov} is then replaced by a stochastic differential equation for $\rho$ driven by  a Poisson or a Wiener process.

Kraus maps are contractions for the trace distance (see~\cite{nielsen-chang-book}, theorem 9.2, page 406):  for all density  matrices $\sigma$, $\rho$, one  has $ \tr{\left|\KK(\sigma)-\KK(\rho) \right|} \leq \tr{|\sigma-\rho|}$.
When $\sigma$ and $\rho$ are pure states (projectors of rank one),  $D(\sigma,\rho)=\sqrt{1-F(\sigma,\rho)}$. Consequently inequality~\eqref{eq:ineq}  yields to
$$
\sum_{\mu=1}^{m}
  \tr{M_\mu \rho M_\mu^\dag}
  D \left(
     \tfrac{M_\mu \sigma M_\mu^\dag }{\tr{M_\mu \sigma M_\mu^\dag}},
      \tfrac{M_\mu\rho M_\mu^\dag }{\tr{M_\mu\rho M_\mu^\dag}}
     \right)
     \leq
      D( \sigma,\rho)
$$
for any pure states  $\sigma$ and $\rho$ (use the fact that $[0,1]\ni x \mapsto \sqrt{1-x}$ is decreasing and concave). It is then tempting to  conjecture that the above  inequality  holds  true for any mixed states $\rho$ and $\sigma$, i.e.,  that
$D(\hat\rho_k,\rho_k)=\tr{\left|\hat\rho_k-\rho_k\right|}$  is  a super-martingale of the Markov chain defined by~\eqref{eq:markov} and~\eqref{eq:Qfilter}. Unfortunately, this is not true in general as shown by the following counter-example.  Take  $n=3$, $m=2$ and
$$
\rho=\begin{pmatrix}
       \tfrac{1}{2} &  0& 0 \\
        0 &  \tfrac{1}{2} & 0 \\
        0  & 0 & 0 \\
     \end{pmatrix}, \quad
     \sigma =
     \begin{pmatrix}
       0 &  0& 0 \\
        0 &  \tfrac{1}{2} & 0 \\
        0  & 0 & \tfrac{1}{2} \\
     \end{pmatrix}, \quad
    M_1 =
     \begin{pmatrix}
       1 &  0& 0 \\
        0 &  \tfrac{1}{\sqrt{2}} & 0 \\
        0  & 0 & 0 \\
     \end{pmatrix}, \quad
    M_2 =
     \begin{pmatrix}
       0 &  0& 0 \\
        0 &  \tfrac{1}{\sqrt{2}} & 0 \\
        0  & 0 & 1 \\
     \end{pmatrix}.
$$
Then  $\sum_{\mu=1}^{2}
  \tr{M_\mu \rho M_\mu^\dag}
  D \left(
     \tfrac{M_\mu \sigma M_\mu^\dag }{\tr{M_\mu \sigma M_\mu^\dag}},
      \tfrac{M_\mu\rho M_\mu^\dag }{\tr{M_\mu\rho M_\mu^\dag}}
     \right)= \tfrac{4}{3} >  1 = D(\sigma,\rho)$.\footnote{
     Notice that we have
     $\sum_{\mu=1}^{2}
  \tr{M_\mu \rho M_\mu^\dag}
  F \left(
     \tfrac{M_\mu \sigma M_\mu^\dag }{\tr{M_\mu \sigma M_\mu^\dag}},
      \tfrac{M_\mu\rho M_\mu^\dag }{\tr{M_\mu\rho M_\mu^\dag}}
     \right)= \tfrac{1}{4}+\tfrac{1}{12}>  \tfrac{1}{4} = F(\sigma,\rho) $. }
     The same counter example shows that the relative entropy
     $\tr{\rho_k \log\rho_k} - \tr{\rho_k \log\hat\rho_k}$ is not, in general,  a super-martingale.


\end{document}